\documentclass[preprint,12pt,10pt,3p]{elsarticle}

\usepackage{amssymb}
\usepackage{amsthm}
\usepackage{graphicx}
\usepackage{amsmath}
\usepackage{hyperref}

\usepackage{pgf,tikz}

\usetikzlibrary{snakes}
\tikzstyle{printersafe}=[snake=snake,segment amplitude=0 pt]

\usetikzlibrary{arrows}
\usetikzlibrary{shapes}

\newtheorem{proposition}{\em Proposition}[section]
\newtheorem{theorem}{\em Theorem}[section]
\newtheorem{conjecture}{\em Conjecture}[section]
\newtheorem{definition}{\em Definition}[section]
\newtheorem{lemma}{\em Lemma}[section]
\newtheorem{remark}{\em Remark}[section]
\newtheorem{corollary}{\em Corollary}[section]

\journal{Sample Journal}

\begin{document}

\begin{frontmatter}
								
	\title{Graphs, Disjoint Matchings and Some Inequalities}
								
	\author{Lianna Hambardzumyan\corref{cor1}}
	\ead{lianna.hambardzumyan@gmail.com}
	\cortext[cor1]{Corresponding author}
								
	\author{Vahan Mkrtchyan}
	\ead{vahanmkrtchyan2002@ysu.am}
								
	\address{Department of Informatics and Applied Mathematics, Yerevan State University, Yerevan, 0025, Armenia}

	\begin{abstract}For $k \geq 1$ and a graph $G$ let $\nu_k(G)$ denote the size of a maximum $k$-edge-colorable subgraph of $G$. Mkrtchyan, Petrosyan and Vardanyan proved that $\nu_2(G)\geq \frac45\cdot |V(G)|$, $\nu_3(G)\geq \frac76\cdot |V(G)|$ for any cubic graph $G$ ~\cite{samvel:2010}. They were also able to show that if $G$ is a cubic graph, then $\nu_2(G)+\nu_3(G)\geq 2\cdot |V(G)|$ ~\cite{samvel:2014} and $\nu_2(G) \leq \frac{|V(G)| + 2\cdot \nu_3(G)}{4}$ ~\cite{samvel:2010}. In the first part of the present work, we show that the last two inequalities imply the first two of them. 
		Moreover, we show that $\nu_2(G) \geq \alpha \cdot \frac{|V(G)| + 2\cdot \nu_3(G)}{4} $, where 
		\\
						
		$\alpha=\frac{16}{17}$, if $G$ is a cubic graph, 
		\\
						
		$\alpha=\frac{20}{21}$, if $G$ is a cubic graph containing a perfect matching, 
		\\
					
		$\alpha=\frac{44}{45}$, if $G$ is a bridgeless cubic graph.
		We also investigate the parameters $\nu_2(G)$ and $\nu_3(G)$ in the class of claw-free cubic graphs. We improve the lower bounds for $\nu_2(G)$ and $\nu_3(G)$ for  claw-free bridgeless cubic graphs to $\nu_2(G)\geq \frac{35}{36}\cdot |V(G)|$ 
		($n \geq 48$), $\nu_3(G)\geq \frac{43}{45}\cdot |E(G)|$. On the basis of these inequalities we are able to improve the  coefficient $\alpha$ for bridgeless claw-free cubic graphs. 
		
		In the second part of the work, we prove lower bounds for $\nu_k(G)$ in terms of $\frac{\nu_{k-1}(G)+\nu_{k+1}(G)}{2}$ for $k\geq 2$ and graphs $G$ containing at most $1$ cycle. We also present the corresponding conjectures for bipartite and nearly bipartite graphs.
	\end{abstract}
								
								
\end{frontmatter}


\section{Introduction}
\label{intro}

In this paper graphs are assumed to be finite, undirected and without loops, though they may contain multi-edges. 

The set of vertices and edges of a graph $G$ will be denoted by $V(G)$ and $E(G)$, respectively. Sometimes we will denote $|V(G)|$ by $n$. The degree of a vertex $u$ of $G$ is denoted by $d_G(u)$. Let $\Delta(G)$ be the maximum degree of a vertex of $G$. A graph is {\em cubic} if every vertex has degree $3$.

A {\em matching} in a graph is a set of edges without common vertices. A matching which covers all vertices of the graph is called a {\em perfect matching}.
A {\em $k$-factor} of a graph is a spanning $k$-regular subgraph. In particular, the edge-set of a $1$-factor is a perfect matching. Moreover, a $2$-factor is a set of cycles in the graph that covers all its vertices.  
We will denote the smallest possible number of odd cycles in a $2$-factor of a cubic graph $G$ by $o(G)$.

A part of this paper works with subclass of cubic graphs, which are called {\em claw-free} cubic graphs. A graph is {\em claw-free} if it has no induced subgraph isomorphic to $K_{1,3}$.

A graph $G$ is called {\em $k$-edge colorable}, if its edges can be assigned $k$ colors so that adjacent edges receive different colors. A subgraph $H$ of a graph $G$ is called {\em maximum $k$-edge-colorable}, if $H$ is $k$-edge-colorable and contains maximum number of edges among all $k$-edge-colorable 
graphs. If $H$ is a $k$-edge-colorable subgraph of $G$ and $e\notin E(H)$, then we will say that $e$ is an {\em uncolored} edge with respect to $H$. If it is clear from the context with respect to which subgraph an edge is uncolored, we will not mention the subgraph.

By a classical result due to Shannon \cite{Shannon:1949,stiebitz:2012,vizing:1964}, we have that cubic graphs are $4$-edge-colorable. It is an interesting and  useful problem to investigate the sizes of subgraphs of cubic graphs that are colorable only with $1$, $2$ or $3$ colors. 

For $k\geq 1$ and a graph $G$ let 
\begin{eqnarray*}
	\nu_k(G)=\max \{|E(H)|: H \textrm{ is a } k\textrm{-edge-colorable subgraph of } G \}.
\end{eqnarray*}

The {\em resistance $r_3(G)$} of a cubic graph $G$ is the minimum of number of edges that have to be removed from $G$ in order to obtain a $3$-edge-colorable graph.  Note that $r_3(G) = |E(G)| - \nu_3(G)$. 

Albertson and Haas \cite{haas:1996,haas:1997}, Steffen \cite{steffen:1998,steffen:2004} and Mkrtchyan et al. \cite{samvel:2010} investigated the lower bounds for $\frac{\nu_k(G)}{|V(G)|}$ in cubic graphs. As a result, in \cite{samvel:2010} an interesting relation between $\nu_2(G)$ and $\nu_3(G)$ 
is proved, which states that for any cubic graph $G$
\begin{eqnarray*}
	\nu_2(G) \leq \frac{|V(G)| + 2\cdot \nu_3(G)}{4}.
\end{eqnarray*} Observe that when $G$ contains a perfect matching ($\nu_1(G)=\frac{|V(G)|}{2}$), in particular, when $G$ is a bridgeless cubic graph, the above-mentioned inequality can be written as
\begin{eqnarray*}
	\nu_2(G) \leq \frac{\nu_1(G) + \nu_3(G)}{2}.
\end{eqnarray*}

The lower bounds for $\frac{\nu_k(G)}{|V(G)|}$ in cubic graphs has been investigated in \cite{bollobas:1978,henning:2007,nishizeki:1981,nishizeki:1979,weinstein:1974} when $k=1$, and for regular graphs of high girth in \cite{flaxman:2007}. This lower bounds has also been investigated in the case when the graphs need not be cubic \cite{miXumbFranciaciq:2013,Kaminski:2014,Rizzi:2009}.

In the present work we give short proofs of main results of Mkrtchyan et. al. ~\cite{samvel:2010}. We also prove lower bounds for $\nu_2(G)$ in terms of $|V(G)|$ and $\frac{|V(G)| + 2\cdot \nu_3(G)}{4}$ in the following sub-classes of cubic graphs: 
\begin{enumerate}
	\item
	      \begin{enumerate}        
	      	\item cubic graphs
	      	\item cubic graphs containing a perfect matching
	      	\item bridgeless cubic graphs
	      \end{enumerate}
	\item
	      \begin{enumerate}
	      	\item claw-free cubic graphs
	      	\item claw-free bridgeless cubic graphs
	      \end{enumerate}
\end{enumerate} In some cases our lower bounds are best-possible. 

In the second part of the work, we investigate $2$ conjectures for bipartite and nearly bipartite graphs. Recall that a graph $G$ is bipartite, if $V(G)$ can be partitioned into $2$ sets $V_1$ and $V_2$, such that any edge of $G$ joins a vertex from $V_1$ to a vertex from $V_2$. $G$ is nearly bipartite, if $G$ contains a vertex $w$, such that $G-w$ is bipartite. Our conjectures state: 

\begin{conjecture}\label{NearlyBipkConj} For any $k\geq 2$ and a nearly bipartite graph $G$, 
\begin{eqnarray*}
	\nu_{k}(G) \geq \left \lfloor \frac{\nu_{k-1}(G) + \nu_{k+1}(G)}{2} \right \rfloor.
\end{eqnarray*}
\end{conjecture}

\begin{conjecture}\label{BipkConj} For any $k\geq 2$ and a bipartite graph $G$, 
\begin{eqnarray*}
	\nu_k(G) \geq \frac{\nu_{k-1}(G) + \nu_{k+1}(G)}{2}.
\end{eqnarray*}
\end{conjecture} Our main results state that these conjectures are true for graphs $G$ containing at most $1$ cycle. 

Terms and concepts that we do not define, can be found in \cite{harary:1969,west:1996}.

\section{Inequalities and bounds for cubic graphs}
\label{app:theorem}


First we formulate a proposition that will be helpful for our presentation of results. It has been applied already for bridgeless cubic graphs in \cite{steffen:2004}. Here we state and prove it for general graphs.

\begin{proposition} \label{theorem:2/3} 
	For any graph $G$ 
	\[ \nu_2(G) \geq \frac{2}{3}\cdot \nu_3(G). \]
\end{proposition}
\begin{proof}
	Let $(H, H', H'')$ be a triple of edge-disjoint matchings of $G$ with $|H|+|H'|+|H''|=\nu_3(G)$ and $|H| \ge |H'| \ge |H''|$. Then $\nu_2(G) \ge |H| + |H'| \ge \frac{2}{3} \nu_3(G). $  
\end{proof}
\noindent

\medskip

In ~\cite{samvel:2010} Mkrtchyan, Petrosyan and Vardanyan proved that
\begin{theorem}\label{theorem:MainResultsCubics} For any cubic graph $G$
	\begin{enumerate}
		\item[(1)] $\nu_2(G)\geq \frac45\cdot |V(G)|$,
		\item[(2)] $\nu_3(G)\geq \frac76\cdot |V(G)|$,
		\item[(3)] $\nu_2(G)+\nu_3(G)\geq 2\cdot |V(G)|$,
		\item[(4)] $\nu_2(G) \leq \frac{|V(G)| + 2\cdot \nu_3(G)}{4}$.
	\end{enumerate}
\end{theorem}

The proofs of (1) and (2) of Theorem \ref{theorem:MainResultsCubics} given in ~\cite{samvel:2010} are long. Here we show that (3) and (4) imply (1) and (2).

\medskip

\begin{theorem}\label{theorem:4/5} For every cubic graph $G$
	\[ \nu_2(G) \geq \frac{4}{5}\cdot |V(G)|. \]
\end{theorem}

\begin{proof} The claim follows immediately by a linear combination of 
inequality (3) of Theorem \ref{theorem:MainResultsCubics} and Proposition \ref{theorem:2/3}: add the former with coefficient $\frac{2}{3}$ to the 
latter.

\end{proof}

\medskip

\begin{theorem}\label{theorem:7/6} For every cubic graph $G$
	\[ \nu_3(G) \geq \frac{7}{6}\cdot |V(G)|. \]
\end{theorem}

\begin{proof} Due to (3) of Theorem \ref{theorem:MainResultsCubics}, we have
	\[ \nu_2(G) + \nu_3(G) \geq 2\cdot |V(G)|. \] 
	(4) of Theorem \ref{theorem:MainResultsCubics} states: 
	\[\nu_2(G) \leq \frac{|V(G)| + 2\cdot \nu_3(G)}{4}.\]
	So, we have:
	\[\frac{|V(G)| + 2\cdot \nu_3(G)}{4} + \nu_3(G) \geq 2\cdot |V(G)|,\]
	or
	\[ |V(G)| + 2\cdot \nu_3(G) + 4\cdot \nu_3(G) \geq 8\cdot |V(G)|, \]
	hence,
	\[\nu_3(G) \geq \frac{7}{6}\cdot |V(G)|.\]
	
	The proof of Theorem \ref{theorem:7/6} is complete.
\end{proof}

\medskip

The following graph on $6$ veritices is a tight example for the inequality in Theorem \ref{theorem:7/6} (Figure ~\ref{fig:tightExample76}).
\begin{figure}[ht]
	\begin{center}
	                        
		\begin{tikzpicture}[line cap=round,line join=round,>=triangle 45,x=1.0cm,y=1.0cm]
			\clip(-3.,-1.) rectangle (3.,1.);
			\draw [line width=1.6pt] (-2.,0.6)-- (-2.,-0.6);
			\draw [line width=1.6pt] (-2.,0.6)-- (-2.,-0.6);
			\draw [line width=1.6pt] (-2.,-0.6)-- (-1.,0.);
			\draw [line width=1.6pt] (-2.,0.6)-- (-1.,0.);
			\draw [line width=1.6pt] (-1.,0.)-- (1.,0.);
			\draw [line width=1.6pt] (1.,0.)-- (2.,0.6);
			\draw [line width=1.6pt] (2.,-0.6)-- (1.,0.);
			\draw [line width=1.6pt] (2.,0.6)-- (2.,-0.6);
			\draw [shift={(-1.69,0.)},line width=1.6pt]  plot[domain=2.0476881944604317:4.2354971127191545,variable=\t]({1.*0.6753517601961218*cos(\t r)+0.*0.6753517601961218*sin(\t r)},{0.*0.6753517601961218*cos(\t r)+1.*0.6753517601961218*sin(\t r)});
			\draw [shift={(1.69,0.)},line width=1.6pt]  plot[domain=-1.0939044591293614:1.0939044591293614,variable=\t]({1.*0.6753517601961218*cos(\t r)+0.*0.6753517601961218*sin(\t r)},{0.*0.6753517601961218*cos(\t r)+1.*0.6753517601961218*sin(\t r)});
			\begin{scriptsize}
				\draw [fill=black] (-2.,-0.6) circle (3pt);
				\draw [fill=black] (-2.,0.6) circle (3pt);
				\draw [fill=black] (-1.,0.) circle (3pt);
				\draw [fill=black] (1.,0.) circle (3pt);
				\draw [fill=black] (2.,0.6) circle (3pt);
				\draw [fill=black] (2.,-0.6) circle (3pt);
			\end{scriptsize}
		\end{tikzpicture}
	\end{center}
	\caption{An example attaining the bound of Theorem \ref{theorem:7/6}.}
	\label{fig:tightExample76}
\end{figure}
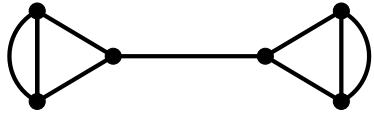

\medskip

\noindent
Inequality (4) of Theorem \ref{theorem:MainResultsCubics} provides an upper bound for $\nu_2(G)$ in terms of $\frac{|V(G)| + 2\cdot \nu_3(G)}{4}$. Here we address the problem of finding a lower bound for $\nu_2(G)$ in terms of the same expression. We investigate this problem in the class of cubic graphs, the class of cubic graphs containing a perfect matching and the class of bridgeless cubic graphs.

Our first result states:

\begin{theorem}\label{theorem:16/17} For any cubic graph $G$
	\[ \nu_2(G) \geq \frac{16}{17}\cdot \frac{|V(G)| + 2\cdot \nu_3(G)}{4}. \]
\end{theorem}
\begin{proof}
The claim follows immediately by a linear combination of 
inequality (3) of Theorem \ref{theorem:MainResultsCubics} and Proposition \ref{theorem:2/3}: add the former with coefficient $\frac{5}{17}$ to the 
latter with coefficient $\frac{12}{17}$.

\end{proof}
\noindent
The Sylvester graph on $10$ vertices is a tight example for this inequality (Figure \ref{fig:tightexample5678}).

\begin{figure}[ht]
	\begin{center}
	\tikzstyle{every node}=[circle, draw, fill=black!50,
                        inner sep=0pt, minimum width=4pt]
                        
		\begin{tikzpicture}
			thick,main node/.style={circle,fill=blue!20,draw,font=\sffamily\Large\bfseries}]
																								
			\node[circle,fill=black,draw] at (-5.5,-1) (n1) {};
																								
			\node[circle,fill=black,draw] at (-6, -0.5) (n2) {};
																								
			\node[circle,fill=black,draw] at (-5,-0.5) (n3) {};
																								
			\node[circle,fill=black,draw] at (-3.5,-1) (n4) {};
																								
			\node[circle,fill=black,draw] at (-4, -0.5) (n5) {};
																								
			\node[circle,fill=black,draw] at (-3,-0.5) (n6) {};
																								
			\node[circle,fill=black,draw] at (-1.5,-1) (n7) {};
																								
			\node[circle,fill=black,draw] at (-2, -0.5) (n8) {};
																								
			\node[circle,fill=black,draw] at (-1,-0.5) (n9) {};
																								
			\node[circle,fill=black,draw] at (-3.5,-2) (n10) {};
																								
			\path[every node]
			(n1) edge  (n2)
																								    
			edge  (n3)
			edge (n10)
																								   	
			(n2) edge (n3)
			edge [bend left] (n3)
																								       
			(n3) 
			(n4) edge (n5)
			edge (n6)
			edge (n10)
																								    
			(n5) edge (n6)
			edge [bend left] (n6)
			(n6)
																								   
			(n7) edge (n8)
			edge (n9)
			edge (n10)
																								    
			(n8) edge (n9)
			edge [bend left] (n9)
																								  
			;
		\end{tikzpicture}
																
	\end{center}
								
	\caption{An example attaining the bound of Theorem \ref{theorem:16/17}.}
	\label{fig:tightexample5678}
\end{figure}
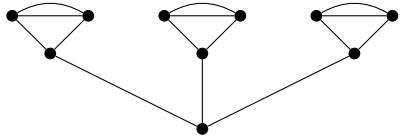
\noindent
\medskip
For cubic graphs containing a perfect matching, we are able to improve the proved lower bound. The proof of this result requires the following auxiliary lemma.

\begin{lemma}\label{lemma:5/6} For any cubic graph $G$ containing a perfect  matching
	\[ \nu_2(G) \geq \frac{5}{6} \cdot n.\]
\end{lemma}
\begin{proof} Let $F$ be a perfect matching of $G$, and let $o(\bar{F})$ be the number of odd cycles in the $2$-factor $G-F$. A $2$-edge-colorable subgraph of $G$ can be obtained by taking $F$ and a maximum matching in $G-F$. Hence, we have
\[\nu_2(G) \geq \frac{n}{2}+\frac{n-o(\bar{F})}{2}.\]
Since the length of each odd cycle of $G-\bar{F}$ is at least $3$, we have
	\[ o(\bar{F}) \leq \frac{n}{3}.\]
	Hence,
	\[ \nu_2(G) \geq \frac{n}{2} + \frac{n}{3}= \frac{5}{6} \cdot n.\]
	
	The proof of Lemma \ref{lemma:5/6} is complete.
\end{proof}

We are ready to prove the main theorem for the class of cubic graphs containing a perfect matching.

\begin{theorem}\label{theorem:20/21} For any cubic graph $G$ containing a perfect  matching
	\[ \nu_2(G) \geq \frac{20}{21}\cdot \frac{|V(G)| + 2\cdot \nu_3(G)}{4}. \]
\end{theorem}
\begin{proof}
The claim follows immediately by a linear combination of 
Lemma \ref{lemma:5/6} and Proposition \ref{theorem:2/3}: add the former with coefficient $\frac{6}{21}$ to the 
latter with coefficient $\frac{15}{21}$.
					
\end{proof}
\noindent
The graph from Figure \ref{fig:tightExample20/21} attains the bound of Theorem \ref{theorem:20/21}.

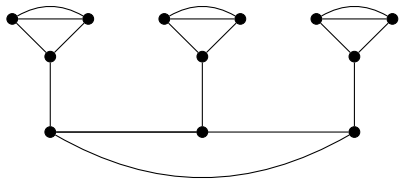
\begin{figure}[ht]
	\begin{center}
	\tikzstyle{every node}=[circle, draw, fill=black!50,
                        inner sep=0pt, minimum width=4pt]
                        
		\begin{tikzpicture}
																								
			\node[circle,fill=black,draw] at (-5.5,-1) (n1) {};
																								
			\node[circle,fill=black,draw] at (-6, -0.5) (n2) {};
																								
			\node[circle,fill=black,draw] at (-5,-0.5) (n3) {};
																								
			\node[circle,fill=black,draw] at (-3.5,-1) (n4) {};
																								
			\node[circle,fill=black,draw] at (-4, -0.5) (n5) {};
																								
			\node[circle,fill=black,draw] at (-3,-0.5) (n6) {};
																								
			\node[circle,fill=black,draw] at (-1.5,-1) (n7) {};
																								
			\node[circle,fill=black,draw] at (-2, -0.5) (n8) {};
																								
			\node[circle,fill=black,draw] at (-1,-0.5) (n9) {};
																								
			\node[circle,fill=black,draw] at (-5.5,-2) (n10) {};
																								
			\node[circle,fill=black,draw] at (-3.5,-2) (n11) {};
																								 
			\node[circle,fill=black,draw] at (-1.5,-2) (n12) {};
																								
			\path[every node]
			(n1) edge  (n2)
																								    
			edge  (n3)
			edge (n10)
																								   	
			(n2) edge (n3)
			edge [bend left] (n3)
																								       
			(n3) 
			(n4) edge (n5)
			edge (n6)
			edge (n11)
																								    
			(n5) edge (n6)
			edge [bend left] (n6)
			(n6)
																								   
			(n7) edge (n8)
			edge (n9)
			edge (n12)
																								    
			(n8) edge (n9)
			edge [bend left] (n9)
																								   
			(n10) edge (n11)
			edge (n12)
																								   
			(n10) edge [bend right] (n12)
																								  
			;
		\end{tikzpicture}
																
	\end{center}
	
	\caption{An example attaining the bound of Theorem \ref{theorem:20/21}.}\label{fig:tightExample20/21}
\end{figure}

\medskip

Petersen theorem states that any bridgeless cubic graph contains a perfect matching \cite{west:1996}. Hence, one can claim that
\[ \nu_2(G) \geq \frac{20}{21}\cdot \frac{|V(G)| + 2\cdot \nu_3(G)}{4} \]
for this class of graphs. It turns out that no bridgeless cubic graph can attain this bound. In other words, we are able to improve the coefficient $\frac{20}{21}$ in this class. 

Our proof will require the following proposition, which is easy to see to be true. It implicitly makes use of the fact, that there is no a bridgeless cubic graph $G$ with $r_3(G)=1$ \cite{steffen:1998,steffen:2004}.

\begin{proposition} \label{prop:r3(G)} Let $G$ be a bridgeless cubic graph. 
\begin{enumerate}
\item [(1)] If $r_3(G)\leq 2$, then \[ \nu_2(G)=\frac{|V(G)| + 2\cdot \nu_3(G)}{4}. \]
\item [(2)] If $r_3(G)$ is odd, then \[ \nu_2(G)<\frac{|V(G)| + 2\cdot \nu_3(G)}{4}. \]
\end{enumerate}
\end{proposition}

Our main result states:

\begin{theorem}\label{thm4445} For any bridgeless cubic graph $G$
	\[ \nu_2(G) \geq \frac{44}{45}\cdot \frac{|V(G)| + 2\cdot \nu_3(G)}{4}. \]
\end{theorem}
\begin{proof} If $n\leq 10$, then it is known that $r_3(G)\leq 2$. Hence, (1) of Proposition \ref{prop:r3(G)} implies that $G$ satisfies the statement of the theorem. Thus without loss of generality, we can assume that $n\geq 12$.

	Steffen in \cite{steffen:2004} proved that $ \nu_2(G) \geq \frac{11}{12}\cdot n $ when $n \geq 12$. Then the claim follows immediately by a linear combination of 
this inequality and Proposition \ref{theorem:2/3}: add the former with coefficient $\frac{12}{45}$ to the 
latter with coefficient $\frac{33}{45}$.

\end{proof}

We are not able to exhibit a bridgeless cubic graph attaining this bound. Moreover, we suspect that

\begin{conjecture}\label{thm5253} For any bridgeless cubic graph $G$
	\[ \nu_2(G) \geq \frac{52}{53}\cdot \frac{|V(G)| + 2\cdot \nu_3(G)}{4}. \]
\end{conjecture}

Using the results of \cite{cavi:1998}, we can show that this conjecture holds for any bridgeless cubic graph with $|V(G)|\leq 26$. In \cite{cavi:1998} it is shown that any connected non-$3$-edge-colorable bridgeless cubic graph $G$ contains a vertex $w$ such that $G-w$ is Hamiltonian. One can easily see that this implies that $r_3(G)\leq 2$. Hence, $\nu_2(G)=\frac{|V(G)| + 2\cdot \nu_3(G)}{4}$ due to (1) of Proposition \ref{prop:r3(G)}.

The coefficient $\frac{52}{53}$ is best-possible in the above conjecture, as the graph from Figure \ref{tightexample6789} attains it. The graph has $28$ vertices and it is constructed as follows: we take $3$ vertex disjoint copies of Petersen graph without a vertex (see left of Figure \ref{tightexample6789}) and connect them according to the right of Figure \ref{tightexample6789}.

\begin{figure}[ht]
	\begin{center}
	\tikzstyle{every node}=[circle, draw, fill=black!50,
                        inner sep=0pt, minimum width=4pt]
                        
		\begin{tikzpicture}
																								
			\node[circle,fill=black,draw] at (0,0) (n1) {};
																								
			\node[circle,fill=black,draw] at (-2.5, 0) (n2) {};
																								
			\node[circle,fill=black,draw] at (2.5,0) (n3) {};
																								
			\node[circle,fill=black,draw] at (-1.5,-0.75) (n4) {};
																								
			\node[circle,fill=black,draw] at (1.5, -0.75) (n5) {};
																								
			\node[circle,fill=black,draw] at (-1,-2) (n6) {};
																								
			\node[circle,fill=black,draw] at (1,-2) (n7) {};
																								
			\node[circle,fill=black,draw] at (-2, -3) (n8) {};
																								
			\node[circle,fill=black,draw] at (2,-3) (n9) {};
																								
			\node[circle,fill=black,draw] at (5,-1) (n10) {};
																								
			\node[circle,fill=black,draw] at (5,-0.25) (n11) {};
																								 
			\node[circle,fill=black,draw] at (4.1,-1.75) (n12) {};
																								  
			\node[circle,fill=black,draw] at (6.9,-1.75) (n13) {};
																								  
			\node[circle,fill=black,draw] at (4.5,0.25) (n14) {};
			\node[circle,fill=black,draw] at (5.5,0.25) (n15) {};
																								  
			\node[circle,fill=black,draw] at (3.4,-2) (n16) {};
			\node[circle,fill=black,draw] at (4.6,-2) (n17) {};
																								  
			\node[circle,fill=black,draw] at (6.3,-2) (n18) {};
			\node[circle,fill=black,draw] at (7.6,-2) (n19) {};

			\path[every node]
			(n1) edge  (n6)
																								    
			edge  (n7)
																								   	
			(n2) edge (n4)
			edge (n8)
																								       
			(n3) edge (n5)
			edge (n9)
			(n4) edge (n5)
			edge (n7)

			(n5) edge (n6)
																								   
			(n6)edge (n8)
																								   
			(n7) edge (n9)

			(n8) edge (n9)
																								   
			(n10) edge (n11)
			edge (n12)
			edge (n13)
																								   
			(n14) edge (n16)
			(n15) edge (n19)
			(n17) edge (n18)

			;
																								
																								
			\draw (5, 0) ellipse (1cm and 0.5cm);
																								
			\draw (4, -2) ellipse (1cm and 0.5cm);
																								
			\draw (7, -2) ellipse (1cm and 0.5cm);

		\end{tikzpicture}
																
	\end{center}
	\caption{An example attaining the bound of Conjecture \ref{thm5253}.}\label{tightexample6789}
\end{figure}
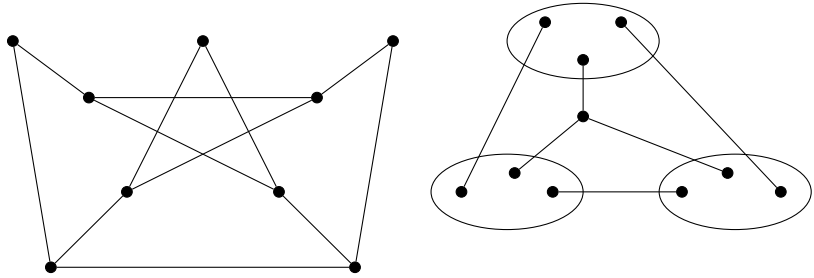

Note that all these coefficients $ \frac{44}{45}$, $\frac{52}{53}$ are very close to $1$, and there are also a vast number of graphs for which mentioned coefficient is $1$, i.e. $ \nu_2(G) = \frac{|V(G)| + 2\cdot \nu_3(G)}{4} $. So, it is an interesting problem to characterize the class of bridgeless cubic graphs $G$ with $ \nu_2(G) = \frac{|V(G)| + 2\cdot \nu_3(G)}{4} $. 

We suspect that

\begin{conjecture}\label{thmnphard} It is {\bf NP-hard} to test whether a given bridgeless cubic graph $G$ satisfies $ \nu_2(G) = \frac{|V(G)| + 2\cdot \nu_3(G)}{4} $.
\end{conjecture}

\section{Inequalities and bounds for claw-free cubic graphs}

This section deals with lower bounds of $\nu_2(G)$ and $\nu_3(G)$ in the class of claw-free cubic graphs. We show that there exist substantial improvements for most of the inequalities proved in the previous section. On the other hand, we demonstrate that some of them cannot be improved.

Before we formulate the new inequalities let us give some definitions.

\begin{definition} A subgraph of $G$ is called a \textbf{diamond} if it is isomorphic to $K_4-e$.
\end{definition}

\begin{definition}
	In graph $G$ a \textbf{string of diamonds} is a maximal sequence $D_1, D_2, . . ., D_k$ of diamonds in which, for every $i\in \{1, 2, . . ., k-1\}$, $D_i$ has a vertex adjacent to a vertex in $D_{i+1}$.
\end{definition}

\begin{definition}
	A connected claw-free cubic graph in which every vertex is in a diamond is called a \textbf{ring of diamonds}.
\end{definition}

\begin{definition}
	\textbf{Replacing a vertex $v$ with
	a triangle} in a cubic graph is to replace $v$ with three vertices $v_1, v_2, v_3$ forming a triangle
	so that if $e_1, e_2, e_3$ are three edges incident with $v$, then $e_1, e_2, e_3$ will be incident with
	$v_1, v_2, v_3$ respectively.
\end{definition} If a graph $G$ is obtained from the graph $H$ by replacing all vertices of $H$ with a triangle, then we will write \textbf{$G = H_\bigtriangleup$}.

We are ready to state the characterization of simple claw-free bridgeless cubic graphs proved by Sang-il Oum in \cite{sang-il_oum:2011}.

\begin{theorem}[\cite{sang-il_oum:2011}]\label{sangIlOum}
	A graph $G$ is a simple $2$-edge-connected claw-free cubic if and only if either
	\begin{itemize}
		\item[(i)] $G$ is isomorphic to $K_4$, or
		\item[(ii)] $G$ is a ring of diamonds, or
		\item[(iii)] $G$ can be built from a 2-edge-connected cubic graph $H$ by replacing some edges of $H$ with strings of diamonds and replacing each vertex of $H$ with a triangle.
	\end{itemize}
	
\end{theorem}

Let us also recall the following classical result of Sumner:

\begin{proposition}\label{sumner}(\cite{sumner:1974})
	If $G$ is a connected claw-free graph of even order, then $G$ has a perfect matching.
\end{proposition}

We are ready to improve the lower bound for $\nu_2(G)$ in the class of claw-free cubic graphs.

\begin{theorem}
	For any claw-free cubic graph $G$ 
	\[ \nu_2(G) \geq \frac{5}{6}\cdot |V(G)|. \]
\end{theorem}
\begin{proof}
	If $G$ is not connected, then by proving the inequality for each connected component we will prove it for $G$. So, we can assume that $G$ is connected. Proposition \ref{sumner} implies that $G$ has a perfect matching. Hence, by Lemma \ref{lemma:5/6}, the above inequality holds.  
\end{proof}

Note that the lower bound on $\nu_3(G)$ for general cubic graphs from Theorem \ref{theorem:7/6} cannot be 
improved for claw-free cubic graphs because the tight example from Figure \ref{fig:tightExample76} is claw-free.

Also, note that the inequality from Theorem \ref{theorem:20/21} cannot be improved for claw-free cubic graphs, as the tight example is a claw-free graph as well. 

\medskip

Below we improve the lower bound for $\nu_3(G)$ in the class of claw-free bridgeless cubic graphs. First, we will state a theorem which holds for every claw-free bridgeless graph, but we will not give a proof of it, as later in the paper we will prove a much better result with a small restriction on the graph size. 
	\begin{theorem}\label{thm2930}
	For any claw-free bridgeless cubic graph $G$
	\[ \nu_2(G) \geq \frac{29}{30}\cdot |V(G)|. \]
\end{theorem}

\begin{theorem}\label{thm4345}
	For any claw-free bridgeless cubic graph $G$ 
	\[ \nu_3(G) \geq \frac{43}{45}\cdot |E(G)|. \]
\end{theorem}
\begin{proof}
Recall that (4) of Theorem \ref{theorem:MainResultsCubics} was stating that $\nu_2(G) \leq \frac{|V(G)| + 2\cdot \nu_3(G)}{4}$. Using this result and Theorem \ref{thm2930} we can easily deduce the statement of this theorem as follows:
\[ \frac{29}{30}|V(G)| \leq \nu_2(G) \leq \frac{|V(G)| + 2\cdot \nu_3(G)}{4}, \]
\[ \frac{29}{30}|V(G)| \leq \frac{|V(G)| + 2\cdot \nu_3(G)}{4}. \]
After some basic calculations we will get
\[ \frac{43}{30} |V(G)| \leq \nu_3(G), \]
which is the same as
\[ \frac{43}{45}|E(G)| \leq \nu_3(G). \]
The proof of the theorem is complete.
\end{proof}

We observe that Theorem \ref{thm4345} is best-possible in a sense that there is a graph attaining the bounds of this theorem. An example of such a graph is $P_\bigtriangleup$, where $P$ is the Petersen graph. 

For the proof of our next result, we will require some lemmas.

\begin{lemma}\label{r3strings} Let $G'$ be a cubic graph, and assume that $G$ is a cubic graph obtained from $G'$ by replacing one of edges of $G'$ with a string of diamonds. Then \[r_3(G)=r_3(G').\]
\end{lemma}

\begin{proof} Assume that the string of diamonds of $G$ that has replaced the edge $a$ of $G'$ contains exactly $k$ diamonds. Then we have
\[|E(G)| = |E(G')|+6k. \]
Taking into account that
\[r_3(G)=|E(G)|-\nu_3(G)\textrm{ and } r_3(G')=|E(G')|-\nu_3(G'),\]
it suffices to show that 
\[\nu_3(G) = \nu_3(G')+6k. \]
It is easy to prove that
\[\nu_3(G) \geq \nu_3(G')+6k, \]
hence, we will only show that
\[\nu_3(G) \leq \nu_3(G')+6k. \]

Let $(H, H', H'')$ be a triple of edge disjoint matchings of $G$, such that their union forms a maximum $3$-edge-colorable subgraph of $G$. Observe that the string itself contains $6k+1$ edges of $G$. Now, if at least one of these edges of $G$ does not belong to $H\cup H'\cup H''$, then the restrictions of these matchings to $G'$ ($H\cap E(G')$, $H'\cap E(G')$, $H''\cap E(G')$) will form a $3$-edge-colorable subgraph of $G'$ (Figure \ref{stringOfDiamonds8}), hence,
\[\nu_3(G')\geq |(H\cup H'\cup H'')\cap E(G')|\geq \nu_3(G)-6k,\]
or 
\[\nu_3(G) \leq \nu_3(G')+6k. \]
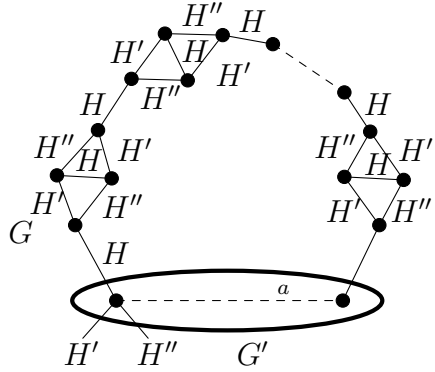
\begin{figure}
		      	      	\begin{center}
		      	      		\begin{tikzpicture}[line cap=round,line join=round,>=triangle 45,x=1.0cm,y=1.0cm]
		      	      			\clip(-3.,-1.) rectangle (3.,4.);
		      	      			\draw [rotate around={0.:(0.,0.)},line width=1.6pt] (0.,0.) ellipse (2.0387377853388053cm and 0.39553983031824425cm);
		      	      			\draw (-1.46,0.)-- (-2.,1.);
		      	      			\draw (-1.46,0.)-- (-1.9,-0.5);
		      	      			\draw (-1.46,0.)-- (-1.04,-0.5);
		      	      			\draw (-2.,1.)-- (-1.52,1.62);
		      	      			\draw (-1.52,1.62)-- (-1.7,2.26);
		      	      			\draw (-1.7,2.26)-- (-2.24,1.66);
		      	      			\draw (-2.24,1.66)-- (-2.,1.);
		      	      			\draw (-2.24,1.66)-- (-1.52,1.62);
		      	      			\draw (-1.7,2.26)-- (-1.26,2.94);
		      	      			\draw (-1.26,2.94)-- (-0.52,2.92);
		      	      			\draw (-0.52,2.92)-- (-0.06,3.52);
		      	      			\draw (-0.06,3.52)-- (-0.82,3.54);
		      	      			\draw (-0.82,3.54)-- (-1.26,2.94);
		      	      			\draw (-0.82,3.54)-- (-0.52,2.92);
		      	      			\draw (1.52,0.)-- (2.,1.);
		      	      			\draw (-0.06,3.52)-- (0.6,3.4);
		      	      			\draw (2.,1.)-- (2.32,1.6);
		      	      			\draw (2.32,1.6)-- (1.88,2.24);
		      	      			\draw (1.88,2.24)-- (1.54,1.64);
		      	      			\draw (1.54,1.64)-- (2.,1.);
		      	      			\draw (1.54,1.64)-- (2.32,1.6);
		      	      			\draw (1.54,2.76)-- (1.88,2.24);
		      	      			\draw [dash pattern=on 3pt off 3pt, printersafe] (-1.46,0.)-- (1.52,0.);
		      	      			\draw [dashed, printersafe] (0.6,3.4)-- (1.54,2.76);
		      	      			\draw (-0.02,-0.4) node[anchor=north west] {$G'$};
		      	      			\draw (-3.02,1.2) node[anchor=north west] {$G$};
		      	      			\draw (-1.8,0.92) node[anchor=north west] {$H$};
		      	      			\draw (-1.8,1.55) node[anchor=north west] {$H''$};
		      	      			\draw (-2.15,2.15) node[anchor=north west] {$H$};
		      	      			\draw (-2.7,2.35) node[anchor=north west] {$H''$};
		      	      			\draw (-2.1,2.9) node[anchor=north west] {$H$};
		      	      			\draw (-0.75,3.6) node[anchor=north west] {$H$};
		      	      			\draw (-1.3,3) node[anchor=north west] {$H''$};
		      	      			\draw (-0.75,4.1) node[anchor=north west] {$H''$};
		      	      			\draw (0.,4) node[anchor=north west] {$H$};
		      	      			\draw (1.65,2.9) node[anchor=north west] {$H$};
		      	      			\draw (1,2.35) node[anchor=north west] {$H''$};
		      	      			\draw (1.65,2.1) node[anchor=north west] {$H$};

		      	      			\draw (-1.6,2.3) node[anchor=north west] {$H'$};
		      	      			\draw (-2.75,1.6) node[anchor=north west] {$H'$};
		      	      			\draw (-0.3,3.3) node[anchor=north west] {$H'$};
		      	      			\draw (-1.7,3.6) node[anchor=north west] {$H'$};
		      	      			\draw (2.1,2.3) node[anchor=north west] {$H'$};
		      	      			\draw (1.15,1.5) node[anchor=north west] {$H'$};
		      	      			\draw (2,1.5) node[anchor=north west] {$H''$};
		      	      			\draw (-2.3,-0.4) node[anchor=north west] {$H'$};
		      	      			\draw (-1.34,-0.4) node[anchor=north west] {$H''$};
		      	      					      	      					      	      			
		      	      			\begin{scriptsize}
		      	      				\draw [fill=black] (-1.46,0.) circle (2.5pt);
		      	      				\draw [fill=black] (1.52,0.) circle (2.5pt);
		      	      				\draw [fill=black] (-2.,1.) circle (2.5pt);
		      	      				\draw [fill=black] (-1.52,1.62) circle (2.5pt);
		      	      				\draw [fill=black] (10.5872,0.) circle (2.5pt);
		      	      				\draw [fill=black] (-1.7,2.26) circle (2.5pt);
		      	      				\draw [fill=black] (-2.24,1.66) circle (2.5pt);
		      	      				\draw [fill=black] (-1.26,2.94) circle (2.5pt);
		      	      				\draw [fill=black] (-0.52,2.92) circle (2.5pt);
		      	      				\draw [fill=black] (-0.06,3.52) circle (2.5pt);
		      	      				\draw [fill=black] (-0.82,3.54) circle (2.5pt);
		      	      				\draw [fill=black] (2.,1.) circle (2.5pt);
		      	      				\draw [fill=black] (0.6,3.4) circle (2.5pt);
		      	      				\draw [fill=black] (2.32,1.6) circle (2.5pt);
		      	      				\draw [fill=black] (1.88,2.24) circle (2.5pt);
		      	      				\draw [fill=black] (1.54,1.64) circle (2.5pt);
		      	      				\draw [fill=black] (1.54,2.76) circle (2.5pt);
		      	      				\draw[color=black] (0.74,0.15) node {$a$};
		      	      			\end{scriptsize}
		      	      		\end{tikzpicture}
		      	      	\end{center}
		      	      	\caption{Restrictions of matchings form a $3$-edge-colorable subgraph of $G'$.}
		      	      	\label{stringOfDiamonds8}
		      	      \end{figure}

Thus, without loss of generality, we can assume that all $6k+1$ edges of the string of $G$ belong to $H\cup H'\cup H''$. Assume that the first edge of the string belongs to $H''$ (Figure \ref{stringOfDiamonds6}). Then, one can easily see that the string should be colored as on Figure \ref{stringOfDiamonds6}. This coloring is unique up to flipping of edges of $H$ and $H'$ in the diamonds of the string.

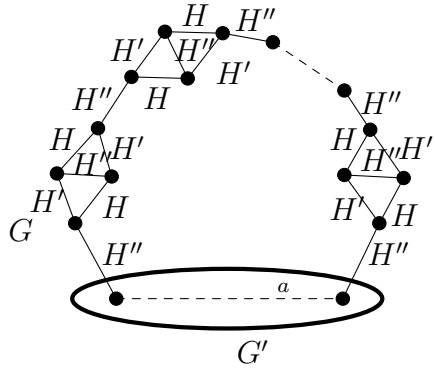
\begin{figure}[ht]
		      	      	\begin{center}
		      	      		\begin{tikzpicture}[line cap=round,line join=round,>=triangle 45,x=1.0cm,y=1.0cm]
		      	      			\clip(-3.,-1.) rectangle (3.,4.);
		      	      			\draw [rotate around={0.:(0.,0.)},line width=1.6pt] (0.,0.) ellipse (2.0387377853388053cm and 0.39553983031824425cm);
		      	      			\draw (-1.46,0.)-- (-2.,1.);
		      	      			\draw (-2.,1.)-- (-1.52,1.62);
		      	      			\draw (-1.52,1.62)-- (-1.7,2.26);
		      	      			\draw (-1.7,2.26)-- (-2.24,1.66);
		      	      			\draw (-2.24,1.66)-- (-2.,1.);
		      	      			\draw (-2.24,1.66)-- (-1.52,1.62);
		      	      			\draw (-1.7,2.26)-- (-1.26,2.94);
		      	      			\draw (-1.26,2.94)-- (-0.52,2.92);
		      	      			\draw (-0.52,2.92)-- (-0.06,3.52);
		      	      			\draw (-0.06,3.52)-- (-0.82,3.54);
		      	      			\draw (-0.82,3.54)-- (-1.26,2.94);
		      	      			\draw (-0.82,3.54)-- (-0.52,2.92);
		      	      			\draw (1.52,0.)-- (2.,1.);
		      	      			\draw (-0.06,3.52)-- (0.6,3.4);
		      	      			\draw (2.,1.)-- (2.32,1.6);
		      	      			\draw (2.32,1.6)-- (1.88,2.24);
		      	      			\draw (1.88,2.24)-- (1.54,1.64);
		      	      			\draw (1.54,1.64)-- (2.,1.);
		      	      			\draw (1.54,1.64)-- (2.32,1.6);
		      	      			\draw (1.54,2.76)-- (1.88,2.24);
		      	      			\draw [dash pattern=on 3pt off 3pt, printersafe] (-1.46,0.)-- (1.52,0.);
		      	      			\draw [dashed, printersafe] (0.6,3.4)-- (1.54,2.76);
		      	      			\draw (-0.02,-0.4) node[anchor=north west] {$G'$};
		      	      			\draw (-3.02,1.2) node[anchor=north west] {$G$};
		      	      			\draw (-1.8,0.9) node[anchor=north west] {$H''$};
		      	      			\draw (-1.8,1.5) node[anchor=north west] {$H$};
		      	      			\draw (-2.2,2.1) node[anchor=north west] {$H''$};
		      	      			\draw (-2.5,2.4) node[anchor=north west] {$H$};
		      	      			\draw (-2.2,2.95) node[anchor=north west] {$H''$};
		      	      			\draw (-1.25,2.95) node[anchor=north west] {$H$};
		      	      			\draw (-0.75,4.05) node[anchor=north west] {$H$};
		      	      			\draw (-0.85,3.6) node[anchor=north west] {$H''$};
		      	      			\draw (-0.05,4) node[anchor=north west] {$H''$};
		      	      			\draw (1.6,2.9) node[anchor=north west] {$H''$};
		      	      			\draw (1.6,2.15) node[anchor=north west] {$H''$};
		      	      			\draw (1.7,0.9) node[anchor=north west] {$H''$};
		      	      			\draw (1.2,2.4) node[anchor=north west] {$H$};
		      	      			\draw (2,1.4) node[anchor=north west] {$H$};
		      	      			\draw (-1.7,2.3) node[anchor=north west] {$H'$};
		      	      			\draw (-2.75,1.6) node[anchor=north west] {$H'$};
		      	      			\draw (-0.3,3.3) node[anchor=north west] {$H'$};
		      	      			\draw (-1.7,3.6) node[anchor=north west] {$H'$};
		      	      			\draw (2.1,2.3) node[anchor=north west] {$H'$};
		      	      			\draw (1.2,1.5) node[anchor=north west] {$H'$};
		      	      			\begin{scriptsize}
		      	      				\draw [fill=black] (-1.46,0.) circle (2.5pt);
		      	      				\draw [fill=black] (1.52,0.) circle (2.5pt);
		      	      				\draw [fill=black] (-2.,1.) circle (2.5pt);
		      	      				\draw [fill=black] (-1.52,1.62) circle (2.5pt);
		      	      				\draw [fill=black] (10.5872,0.) circle (2.5pt);
		      	      				\draw [fill=black] (-1.7,2.26) circle (2.5pt);
		      	      				\draw [fill=black] (-2.24,1.66) circle (2.5pt);
		      	      				\draw [fill=black] (-1.26,2.94) circle (2.5pt);
		      	      				\draw [fill=black] (-0.52,2.92) circle (2.5pt);
		      	      				\draw [fill=black] (-0.06,3.52) circle (2.5pt);
		      	      				\draw [fill=black] (-0.82,3.54) circle (2.5pt);
		      	      				\draw [fill=black] (2.,1.) circle (2.5pt);
		      	      				\draw [fill=black] (0.6,3.4) circle (2.5pt);
		      	      				\draw [fill=black] (2.32,1.6) circle (2.5pt);
		      	      				\draw [fill=black] (1.88,2.24) circle (2.5pt);
		      	      				\draw [fill=black] (1.54,1.64) circle (2.5pt);
		      	      				\draw [fill=black] (1.54,2.76) circle (2.5pt);
		      	      				\draw[color=black] (0.74,0.15) node {$a$};
		      	      			\end{scriptsize}
		      	      		\end{tikzpicture}
		      	      	\end{center}
		      	      	\caption{All edges of the string belong to the matchings.}
		      	      	\label{stringOfDiamonds6}
		      	      \end{figure}

Consider the restrictions of matchings of $H$, $H'$ and $H''$ to $G'$, and add the edge $a$ to $H''$ (Figure \ref{stringOfDiamonds6}). Observe that these new matchings will form a $3$-edge-colorable subgraph of $G'$, hence,
\[\nu_3(G')\geq |(H\cup H'\cup H'')\cap E(G')|= (\nu_3(G)+1)-6k-1=\nu_3(G)-6k,\]
or 
\[\nu_3(G) \leq \nu_3(G')+6k. \]
The proof of the lemma is complete.
\end{proof}

\begin{lemma}\label{r3triangles} (See the proof of Lemma 3.4 from \cite{steffen:2004}) Let $G'$ be a bridgeless cubic graph, and assume that $G$ is a bridgeless cubic graph obtained from $G'$ by replacing one of vertices of $G'$ with a triangle. Then \[r_3(G)=r_3(G').\] 
\end{lemma}

\begin{lemma}\label{lemma:eckhard}\cite{steffen:2004}
	If $G$ is a bridgeless cubic graph with at least $16$ vertices, then \[r_3(G) \leq \frac{|V(G)|}{8}.\]
\end{lemma}

\begin{lemma}\label{lemma:|V|/24} Let $G$ be any claw-free bridgeless cubic graph with $n \geq 48$. Then \[r_3(G)\leq \frac{|V(G)|}{24}.\]
\end{lemma}
\begin{proof} If $r_3(G) \leq 2$, then 
\[r_3(G)\leq 2\leq  \frac{|V(G)|}{24}.\]
Thus, without loss of generality, we can assume that $r_3(G) \geq 3$. If $G$ contains multi-edges, then repeatedly remove the vertices of $G$ adjacent to multi-edges and join the $2$ degree-two vertices with an edge (Figure \ref{multigraph2}). We claim that the resulting graph $G'$ contains no multi-edges and $|V(G')| \geq 84$. 

	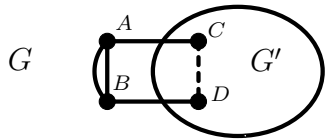
\begin{figure}[ht]
		\begin{center}
			\definecolor{ffffff}{rgb}{1.,1.,1.}
			\begin{tikzpicture}[line cap=round,line join=round,>=triangle 45,x=1.0cm,y=1.0cm]
				\clip(-3.,-1.) rectangle (4.,1.);
				\draw [rotate around={0.:(0.7165454545454546,0.)},line width=1.6pt] (0.7165454545454546,0.) ellipse (1.116643549858145cm and 0.8564201241271993cm);
				\draw [line width=1.6pt] (-1.,0.4)-- (0.2,0.4);
				\draw [line width=1.6pt] (-1.,-0.4)-- (0.2,-0.4);
				\draw [line width=1.6pt] (-1.,0.4)-- (-1.,-0.4);
				\draw [shift={(-0.6069090909090916,0.)},line width=1.6pt]  plot[domain=2.3474831103790965:3.9357021968004897,variable=\t]({1.*0.5608212396208945*cos(\t r)+0.*0.5608212396208945*sin(\t r)},{0.*0.5608212396208945*cos(\t r)+1.*0.5608212396208945*sin(\t r)});
				\draw [line width=1.6pt,dashed, printersafe] (0.2,0.4)-- (0.2,-0.4);
				\draw (0.735549973811124,0.45801048113376025) node[anchor=north west] {$G'$};
				\draw (-1.835549973811124,0.45801048113376025) node[anchor=north east] {$G$};
				\begin{scriptsize}
					\draw [fill=ffffff] (0.8149090909090907,-0.8530909090909052) circle (.5pt);
					\draw[color=ffffff] (0.9658774507820388,-0.6442710157984739) node {$A_2$};
					\draw [fill=black] (0.2,0.4) circle (3pt);
					\draw[color=black] (0.4394146462770906,0.5567222569784379) node {$C$};
					\draw [fill=black] (0.2,-0.4) circle (3pt);
					\draw[color=black] (0.4887705341994295,-0.29877980034210194) node {$D$};
					\draw [fill=black] (-1.,0.4) circle (3pt);
					\draw[color=black] (-0.7780305891406021,0.638982070182336) node {$A$};
					\draw [fill=black] (-1.,-0.4) circle (3pt);
					\draw[color=black] (-0.8109345144221614,-0.15071213657508545) node {$B$};
				\end{scriptsize}
			\end{tikzpicture}
		\end{center}
		\caption{A multi-edge in $G$.}
		\label{multigraph2}
	\end{figure}

If it contains a multi-edge, then one can easily see that $r_3(G)=0$ ($G$ is $3$-edge-colorable), which violates our assumption that $r_3(G) \geq 3$. Hence, $G'$ is simple. Consider the Theorem  \ref{sangIlOum}. As $r_3(G) \geq 3$, we have that the theorem works from point (iii). Let $H$ be the corresponding $2$-edge-connected graph $H$. Lemmas \ref{r3strings} and \ref{r3triangles} imply that $3\leq r_3(G')=r_3(H)$.

Let us show that $|V(H)|\geq 28$. If $|V(H)|\leq 26$, then \cite{cavi:1998} implies that there is a vertex $w$ of $H$ such that $H-w$ is Hamiltonian. One can easily see that this implies that $r_3(H)\leq 2$ contradicting our assumption. Hence, $|V(H)|\geq 28$, which implies that $|V(G')|\geq 3\cdot 28= 84$. 

Thus, without loss of generality, we can assume that our initial graph $G$ is simple.

Similarly, one can show that $G$ contains no string of diamonds. Thus, due to Theorem \ref{sangIlOum}, there is a $2$-edge-connected graph $H$ such that $G=H_\bigtriangleup$. As $|V(G)|\geq 48$, we have $|V(H)|\geq 16$, hence, due to Lemma \ref{lemma:eckhard}, we have
\[r_3(G)=r_3(H)\leq \frac{|V(H)|}{8}=\frac{|V(H_\bigtriangleup)|}{24}= \frac{|V(G)|}{24}.\]
The proof of the lemma is complete.
\end{proof}

\begin{theorem}\label{theorem:3536}
	For any claw-free bridgeless cubic graph $G$ with $n \geq 48$
	\[ \nu_2(G) \geq \frac{35}{36}\cdot |V(G)|. \]
\end{theorem}
\begin{proof} Due to Lemma \ref{lemma:|V|/24}
\[ \nu_3(G) =|E(G)|-r_3(G) \geq \frac{3 \cdot |V(G)|}{2} - \frac{|V(G)|}{24} = \frac{35 \cdot |V(G)|}{24}, \]
hence,
\[ \nu_2(G) \geq \frac23 \cdot \nu_3(G) \geq \frac{2}{3} \cdot \frac{35|V(G)|}{24} = 35\cdot\frac{|V(G)|}{36},  \]
when $|V(G)|\geq 48$.

The proof of the theorem is complete.
\end{proof}

\begin{theorem}
	For any claw-free bridgeless cubic graph $G$ with $n \geq 48$ 
	\[ \nu_2(G) \geq \frac{140}{141}\cdot \frac{|V(G)| + 2\cdot \nu_3(G)}{4}. \]
\end{theorem}
\begin{proof}
	From Theorem \ref{theorem:3536} we have 
	\[ 36\cdot \nu_2(G) \geq 35\cdot n, \]
	hence,
	\[ 141\cdot \nu_2(G) \geq 35\cdot n + 105 \cdot \nu_2(G) = 35\cdot n + 70 \cdot \frac{3}{2} \cdot \nu_2(G) \geq 35 \cdot n + 70 \cdot \nu_3(G). \]
	The last inequality follows from Proposition \ref{theorem:2/3}. Then,
	\[ 141 \cdot \nu_2(G)  \geq 35 \cdot (n + 2 \nu_3(G)).\]
	The final result we can write in the following form:
	\[ \nu_2(G) \geq \frac{140}{141}\cdot \frac{n + 2\cdot \nu_3(G)}{4}. \]
	The proof of the theorem is complete.
\end{proof}

We were unable to find a claw-free bridgeless cubic graph attaining the bound of the previous theorem. Moreover, we suspect that

\begin{conjecture}
	For any claw-free bridgeless cubic graph $G$
	\[ \nu_2(G) \geq \frac{164}{165}\cdot \frac{|V(G)| + 2\cdot \nu_3(G)}{4}. \]
\end{conjecture}

The bound presented by the previous conjecture is tight, in a sense, that there is a graph attaining it. That example is obtained from the graph from Figure \ref{tightexample6789} by replacing all its vertices with triangles.

\section{Inequalities and bounds for trees and unicyclic graphs}
\label{sec:TreeUnicyclic}

In the previous $2$ sections, we have presented some bounds for $\nu_{2}(G)$ in terms of $\frac{|V|+2\cdot \nu_3(G)}{4}$. If a graph $G$ has a perfect matching ($\nu_{1}(G)=\frac{|V|}{2}$), we have
\[\frac{|V|+2\cdot \nu_3(G)}{4}=\frac{\nu_{1}(G)+\nu_{3}(G)}{2}.\]
Hence, one may wonder whether a bound for $\nu_{2}(G)$ can be proved in terms of $\frac{\nu_{1}(G)+\nu_{3}(G)}{2}$ in some interesting graph classes. In the present section, we address a generalization of this question. More precisely, we aim to bound $\nu_{k}(G)$ in terms of $\frac{\nu_{k-1}(G)+\nu_{k+1}(G)}{2}$ for $k\geq 2$. Conjectures \ref{NearlyBipkConj} and \ref{BipkConj} present the main statements that we have tried to prove. In the present section, we verify these conjectures for trees and unicyclic graphs (graphs containing exactly $1$ cycle).

First, we prove some lemmas that will be helpful later in the section.

\begin{lemma}\label{lemma:pendantedge} Let $G$ be a graph, and let $e=(u,v)\in E(G)$. Assume that $d_G(u)=1$. Then for any $k\geq 1$, there is a maximum $k$-edge-colorable subgraph $H_k$ of $G$, such that $e\in E(H_k)$.
\end{lemma}

\begin{proof} Let $H_k$ be any maximum $k$-edge-colorable subgraph of $G$. If $e\in E(H_k)$, then we are done. Thus, we can assume that $e\notin E(H_k)$. Since $H_k$ is a maximum $k$-edge-colorable subgraph of $G$ and $d_G(u)=1$, there is an edge $e'\in E(H_k)$, such that $e'$ is incident to $v$. Consider the subgraph $H'_{k}$ of $G$ defined as follows: $E(H'_{k})=(E(H_{k})\backslash \{e'\})\cup \{e\}$. Observe that $H'_{k}$ is $k$-edge-colorable, $e\in E(H'_k)$ and $|E(H'_{k})|=|E(H_{k})|$, hence $H'_{k}$ is a maximum $k$-edge-colorable subgraph of $G$ containing $e$. The proof of the lemma is complete.
\end{proof}

\begin{lemma}\label{lemma:bridgeedge} Let $k\geq 1$, $G$ be a connected graph, and let $e=(u,v)\in E(G)$ be a bridge of $G$. Assume that there is a maximum $k$-edge-colorable subgraph $H_k$ of $G$, such that $e\in E(H_k)$. Then 
\begin{equation*}
    \nu_{k}(G)=\nu_{k}(G_{1}e)+\nu_{k}(G_{2}e)-1.
\end{equation*} Here $G_1$ and $G_2$ are the components of $G-e$, and $G_{1}e$, $G_{2}e$ are the supergraphs of $G_1$ and $G_2$, respectively, that satisfy the equalities $E(G_{1}e)=E(G_{1})\cup \{e\}$ and $E(G_{2}e)=E(G_{2})\cup \{e\}$.
\end{lemma}

\begin{proof} Let $H^{(1)}$ and $H^{(2)}$ be the restrictions of $H_{k}$ in the graphs $G_{1}e$ and $G_{2}e$, respectively. Clearly, these subgraphs are $k$-edge-colorable. We claim that $H^{(1)}$ and $H^{(2)}$ are maximum $k$-edge-colorable subgraphs of $G_{1}e$ and $G_{2}e$, respectively. Assume that $|E(H^{(1)})|<\nu_{k}(G_{1}e)$. Then, by Lemma \ref{lemma:pendantedge}, there is a maximum $k$-edge-colorable subgraph $H'^{(1)}$ containing $e$. Consider the subgraph $H'_{k}$ of $G$ defined as follows:
\begin{equation*}
    E(H'_{k})=(E(H_k)\backslash E(H^{(1)}))\cup E(H'^{(1)}).
\end{equation*} Observe that $H'_{k}$ is $k$-edge-colorable and $|E(H'_{k})|>|E(H_{k})|$ contradicting the choice of $H_{k}$. Similarly, one can prove that $H^{(2)}$ is a maximum $k$-edge-colorable subgraphs of $G_{2}e$. 

We have the following chain of equalities:
\begin{align*}
    \nu_{k}(G) &=|E(H_{k})|\\
               &= |E(H^{(1)})|+|E(H^{(2)})|-1\\
               &= \nu_{k}(G_{1}e)+\nu_{k}(G_{2}e)-1.
\end{align*}
The proof of the lemma is complete.
\end{proof}

Our first theorem in this section verifies Conjecture \ref{NearlyBipkConj} for graphs with at most $1$ cycle.

\begin{theorem}\label{thm:TreeUniLowBound} For any $k\geq 2$ and a graph $G$ containing at most $1$ cycle, 
\begin{eqnarray*}
	\nu_{k}(G) \geq \left \lfloor \frac{\nu_{k-1}(G) + \nu_{k+1}(G)}{2} \right \rfloor.
\end{eqnarray*}
\end{theorem}

\begin{proof} We will prove the theorem by induction on $n$. The statement of the theorem is trivial when $n=1,2$. Assume that it is true for all graphs having at most $1$ cycle and less than $n$ vertices, and consider a graph $G$ with $n$ vertices and containing at most $1$ cycle. Clearly, we can assume that $G$ is connected and $\Delta(G)\geq 3$ (the statement of the theorem is true for cycles and paths). 

Let $T$ be a tree defined as follows: if $G$ is a tree, then $T=G$, otherwise $T=G/C$. Here $C$ is the only cycle of $G$, and $T$ is the tree obtained from $G$ by contracting $C$ to a vertex $v_C$. View $T$ as a rooted tree. The root of $T$ is any of its vertices, if $G=T$, and is the vertex $v_C$, otherwise. Below, we will speak about children, grand-children of vertices of $G$. This relationship will be viewed from the perspective of the tree $T$.

Let us show that, without loss of generality, we can assume that there is no vertex of $G$ with degree $2$ that is adjacent to a vertex of degree $1$. On the opposite assumption, consider a vertex $z$ of degree $2$ that is adjacent to a vertex $y$ of degree $1$. Observe that since $k\geq 2$, we have $\nu_{i}(G)=1+\nu_{i}(G-y)$ for $i=k,k+1$ and $\nu_{k-1}(G)\leq 1+\nu_{k-1}(G-y)$. Thus, we will have:
\begin{align*}
    \nu_{k}(G) &=\nu_k(G-y)+1\\
               & \geq \left \lfloor \frac{\nu_{k+1}(G-y)  + \nu_{k-1}(G-y) }{2}\right \rfloor+1\\
               & = \left \lfloor \frac{\nu_{k+1}(G-y)+1  + \nu_{k-1}(G-y)+1 }{2}\right \rfloor\\
               &\geq \left \lfloor \frac{\nu_{k+1}(G) + \nu_{k-1}(G)}{2}\right \rfloor.
\end{align*} Here the inequality follows from induction hypothesis applied to $G-y$. Thus, we can assume that no vertex of $G$ that has degree $2$ is adjacent to a vertex of degree $1$.

Next, let us show that all vertices of $G$ with degree at least $3$ lie on $C$-the unique cycle of $G$. On the opposite assumption, consider a vertex $x$ of degree at least $3$ that does not lie on the cycle and it has no children, grand-children, etc. that are of degree at least $3$. Observe that all the children of $x$ are of degree $1$. We will consider some cases.

\medskip

\begin{figure}[ht]
\centering
\begin{minipage}[b]{.5\textwidth}
  \begin{center}
		\begin{tikzpicture}[line cap=round,line join=round,>=triangle 45,x=1.0cm,y=1.0cm]
\clip(-3.5,0.) rectangle (4.,6.000401548929543);
\draw(0.,2.) circle (1.6991762710207556cm);
\draw (2.163243823759741,5.749412616936055) node[anchor=north west] {$\geq k+1$};
\draw (0.9790737243720234,3.388745708272143)-- (1.5377108546082783,4.10310351880089);
\draw (1.5377108546082783,4.10310351880089)-- (1.3475677243101816,4.9155332573473025);
\draw (1.5377108546082783,4.10310351880089)-- (1.7624254631423926,4.898247518229294);
\draw (1.5377108546082783,4.10310351880089)-- (2.385726538800657,4.65615805122504);
\draw (1.5377108546082783,4.10310351880089)-- (2.5518163801005547,4.39012372336622);
\draw (1.6056983061153253,4.1) node[anchor=north west] {$x$};
\draw (-0.22115103017838977,2.423809268022345) node[anchor=north west] {$G'$};
\begin{scriptsize}
\draw [fill=black] (0.9790737243720234,3.388745708272143) circle (2.5pt);
\draw [fill=black] (1.3475677243101816,4.9155332573473025) circle (1.5pt);
\draw [fill=black] (1.7624254631423926,4.898247518229294) circle (1.5pt);
\draw [fill=black] (2.5518163801005547,4.39012372336622) circle (1.5pt);
\draw [fill=black] (1.9667765830359647,4.880274665606372) circle (0.5pt);
\draw [fill=black] (2.2453485331226783,4.762981212938282) circle (0.5pt);
\draw [fill=black] (2.1060625580793215,4.836289620855839) circle (0.5pt);
\draw [fill=black] (1.5377108546082783,4.10310351880089) circle (2.5pt);
\draw [fill=black] (2.385726538800657,4.65615805122504) circle (1.5pt);
\end{scriptsize}
\end{tikzpicture}
		\end{center}
		\caption{$d_G(x) \geq k+2$}
		\label{not_on_cycle_k+2}
\end{minipage}%
\begin{minipage}[b]{.5\textwidth}
  
  \begin{center}
\begin{tikzpicture}[line cap=round,line join=round,>=triangle 45,x=1.0cm,y=1.0cm]
\clip(-3.5,0.) rectangle (4.,6.000401548929543);
\draw(0.,2.) circle (1.6991762710207556cm);
\draw (2.163243823759741,5.749412616936055) node[anchor=north west] {$k$};
\draw (0.9790737243720234,3.388745708272143)-- (1.5377108546082783,4.10310351880089);
\draw (1.5377108546082783,4.10310351880089)-- (1.3475677243101816,4.9155332573473025);
\draw (1.5377108546082783,4.10310351880089)-- (1.7624254631423926,4.898247518229294);
\draw (1.5377108546082783,4.10310351880089)-- (2.385726538800657,4.65615805122504);
\draw (1.5377108546082783,4.10310351880089)-- (2.5518163801005547,4.39012372336622);
\draw (1.3056983061153253,3.8) node[anchor=north west] {$e$};
\draw (1.6056983061153253,4.1) node[anchor=north west] {$x$};
\draw (-0.22115103017838977,2.423809268022345) node[anchor=north west] {$G'$};
\draw (2.644305943413925,4.5362994456341985) node[anchor=north west] {$E'$};
\begin{scriptsize}
\draw [fill=black] (0.9790737243720234,3.388745708272143) circle (2.5pt);
\draw [fill=black] (1.3475677243101816,4.9155332573473025) circle (1.5pt);
\draw [fill=black] (1.7624254631423926,4.898247518229294) circle (1.5pt);
\draw [fill=black] (2.5518163801005547,4.39012372336622) circle (1.5pt);
\draw [fill=black] (1.9667765830359647,4.880274665606372) circle (0.5pt);
\draw [fill=black] (2.2453485331226783,4.762981212938282) circle (0.5pt);
\draw [fill=black] (2.1060625580793215,4.836289620855839) circle (0.5pt);
\draw [fill=black] (1.5377108546082783,4.10310351880089) circle (2.5pt);
\draw [fill=black] (2.385726538800657,4.65615805122504) circle (1.5pt);
\end{scriptsize}
\end{tikzpicture}
		\end{center}
		\caption{$d_G(x) = k+1$}
		\label{not_on_cycle_k+1}
\end{minipage}
\end{figure}

Case 1: $d_G(x) \geq k+2$. Then $G$ can be represented as on Figure \ref{not_on_cycle_k+2}. 

It can be easily seen that in this case there is an edge $e$ adjacent to $x$ such that $\nu_{i}(G)=\nu_{i}(G-e)$ for $i=k-1,k,k+1$. Hence, we have:
\begin{align*}
    \nu_{k}(G) &=\nu_k(G-e) \\
               & \geq \left \lfloor \frac{\nu_{k+1}(G-e)  + \nu_{k-1}(G-e) }{2}\right \rfloor \\
               &=\left \lfloor \frac{\nu_{k+1}(G) + \nu_{k-1}(G) }{2}\right \rfloor.
\end{align*} Here the inequality follows from induction hypothesis applied to the components of $G-e$.

\medskip

Case 2: $3\leq d_G(x) = k+1$. Then $G$ can be represented as on Figure \ref{not_on_cycle_k+1}. Here $E'$ denotes the edge-set of the component of $G-e$ containing $x$.

We have
\[\nu_{k-1}(G) \leq \nu_{k-1}(G') + |E'|-1,\]
\[\nu_{k}(G) = \nu_{k}(G') + |E'|,\]
\[\nu_{k+1}(G) = \nu_{k+1}(G'e) + |E'|.\]
It is easy to see that $\nu_{k+1}(G'e)\leq \nu_{k+1}(G')+1$, hence, by induction hypothesis, we have
\[\left \lfloor \frac{\nu_{k-1}(G') + \nu_{k+1}(G'e)-1}{2} \right \rfloor \leq \left \lfloor \frac{\nu_{k-1}(G') + \nu_{k+1}(G')}{2} \right \rfloor \leq \nu_k(G').\]
The last inequality, in its turn, implies:
\begin{align*}
    \nu_{k}(G) &=\nu_k(G')+ |E'| \\
               & \geq \left \lfloor \frac{\nu_{k-1}(G') + \nu_{k+1}(G'e)-1}{2} \right \rfloor+ |E'| \\
               &\geq \left \lfloor \frac{\nu_{k+1}(G) + \nu_{k-1}(G) }{2}\right \rfloor.
\end{align*} 

\medskip

\begin{figure}[ht]
\centering
\begin{minipage}[b]{.5\textwidth}
  \begin{center}
\begin{tikzpicture}[line cap=round,line join=round,>=triangle 45,x=1.0cm,y=1.0cm]
\clip(-3.5,0.) rectangle (4.,6.000401548929543);
\draw(0.,2.) circle (1.6991762710207556cm);
\draw (2.163243823759741,5.749412616936055) node[anchor=north west] {$k-1$};
\draw (0.9790737243720234,3.388745708272143)-- (1.5377108546082783,4.10310351880089);
\draw (1.5377108546082783,4.10310351880089)-- (1.3475677243101816,4.9155332573473025);
\draw (1.5377108546082783,4.10310351880089)-- (1.7624254631423926,4.898247518229294);
\draw (1.5377108546082783,4.10310351880089)-- (2.385726538800657,4.65615805122504);
\draw (1.5377108546082783,4.10310351880089)-- (2.5518163801005547,4.39012372336622);
\draw (1.3056983061153253,3.8) node[anchor=north west] {$e$};
\draw (1.6056983061153253,4.1) node[anchor=north west] {$x$};
\draw (-0.22115103017838977,2.423809268022345) node[anchor=north west] {$G'$};
\draw (2.644305943413925,4.5362994456341985) node[anchor=north west] {$E'$};
\begin{scriptsize}
\draw [fill=black] (0.9790737243720234,3.388745708272143) circle (2.5pt);
\draw [fill=black] (1.3475677243101816,4.9155332573473025) circle (1.5pt);
\draw [fill=black] (1.7624254631423926,4.898247518229294) circle (1.5pt);
\draw [fill=black] (2.5518163801005547,4.39012372336622) circle (1.5pt);
\draw [fill=black] (1.9667765830359647,4.880274665606372) circle (0.5pt);
\draw [fill=black] (2.2453485331226783,4.762981212938282) circle (0.5pt);
\draw [fill=black] (2.1060625580793215,4.836289620855839) circle (0.5pt);
\draw [fill=black] (1.5377108546082783,4.10310351880089) circle (2.5pt);
\draw [fill=black] (2.385726538800657,4.65615805122504) circle (1.5pt);
\end{scriptsize}
\end{tikzpicture}
		\end{center}
		\caption{$d_G(x) = k$}
		\label{not_on_cycle_k}
\end{minipage}%
\begin{minipage}[b]{.5\textwidth}
  
  \begin{center}
\begin{tikzpicture}[line cap=round,line join=round,>=triangle 45,x=1.0cm,y=1.0cm]
\clip(-3.5,0.) rectangle (4.,6.000401548929543);
\draw(0.,2.) circle (1.6991762710207556cm);
\draw (2.163243823759741,5.749412616936055) node[anchor=north west] {$\leq k-2$};
\draw (0.9790737243720234,3.388745708272143)-- (1.5377108546082783,4.10310351880089);
\draw (1.5377108546082783,4.10310351880089)-- (1.3475677243101816,4.9155332573473025);
\draw (1.5377108546082783,4.10310351880089)-- (1.7624254631423926,4.898247518229294);
\draw (1.5377108546082783,4.10310351880089)-- (2.385726538800657,4.65615805122504);
\draw (1.5377108546082783,4.10310351880089)-- (2.5518163801005547,4.39012372336622);
\draw (1.3056983061153253,3.8) node[anchor=north west] {$e$};
\draw (1.6056983061153253,4.1) node[anchor=north west] {$x$};
\draw (-0.22115103017838977,2.423809268022345) node[anchor=north west] {$G'$};
\draw (2.644305943413925,4.5362994456341985) node[anchor=north west] {$E'$};
\begin{scriptsize}
\draw [fill=black] (0.9790737243720234,3.388745708272143) circle (2.5pt);
\draw [fill=black] (1.3475677243101816,4.9155332573473025) circle (1.5pt);
\draw [fill=black] (1.7624254631423926,4.898247518229294) circle (1.5pt);
\draw [fill=black] (2.5518163801005547,4.39012372336622) circle (1.5pt);
\draw [fill=black] (1.9667765830359647,4.880274665606372) circle (0.5pt);
\draw [fill=black] (2.2453485331226783,4.762981212938282) circle (0.5pt);
\draw [fill=black] (2.1060625580793215,4.836289620855839) circle (0.5pt);
\draw [fill=black] (1.5377108546082783,4.10310351880089) circle (2.5pt);
\draw [fill=black] (2.385726538800657,4.65615805122504) circle (1.5pt);
\end{scriptsize}
\end{tikzpicture}
		\end{center}
		\caption{$d_G(x) \leq k-1$}
		\label{not_on_cycle_k-1}
\end{minipage}
\end{figure}

Case 3: $3\leq d_G(x) = k$. Then $G$ can be represented as on Figure \ref{not_on_cycle_k}. Here $E'$ denotes the edge-set of the component of $G-e$ containing $x$.
	
We have the following equalities:
\[\nu_{k-1}(G) = \nu_{k-1}(G') + |E'|,\]
\[\nu_{k}(G) = \nu_{k}(G'e) + |E'|,\]
\[\nu_{k+1}(G) = \nu_{k+1}(G'e) + |E'|.\]
It is easy to see that $\nu_{k-1}(G')\leq \nu_{k-1}(G'e)$, hence, by induction hypothesis, we have
\[\left \lfloor  \frac{\nu_{k-1}(G') + \nu_{k+1}(G'e)}{2} \right \rfloor \leq \left \lfloor \frac{\nu_{k-1}(G'e) + \nu_{k+1}(G'e)}{2}\right \rfloor  \leq \nu_k(G'e).\]
The last inequality, in turn, implies:
\begin{align*}
    \nu_{k}(G) &=\nu_k(G'e)+ |E'| \\
               & \geq \left \lfloor \frac{\nu_{k-1}(G') + \nu_{k+1}(G'e)}{2} \right \rfloor+ |E'| \\
               &=\left \lfloor \frac{\nu_{k+1}(G) + \nu_{k-1}(G) }{2}\right \rfloor.
\end{align*} 

\medskip

Case 4: $3\leq d_G(x) \leq k-1$. Then $G$ can be represented as on Figure \ref{not_on_cycle_k-1}. Here $E'$ denotes the edge-set of the component of $G-e$ containing $x$.
We have the following equalities:
\[\nu_{k-1}(G) = \nu_{k-1}(G'e) + |E'|,\]
\[\nu_{k}(G) = \nu_{k}(G'e) + |E'|,\]
\[\nu_{k+1}(G) = \nu_{k+1}(G'e) + |E'|.\]
By induction hypothesis, we have
\[\nu_k(G'e) \geq \left \lfloor \frac{\nu_{k-1}(G'e) + \nu_{k+1}(G'e)}{2} \right \rfloor.\]
Hence,
\begin{align*}
    \nu_{k}(G) &=\nu_k(G'e) + |E'| \\
               & \geq \left \lfloor \frac{\nu_{k+1}(G'e)  + \nu_{k-1}(G'e) }{2}\right \rfloor +|E'| \\
               &= \left \lfloor \frac{\nu_{k+1}(G'e) +|E'| + \nu_{k-1}(G'e)+|E'| }{2}\right \rfloor\\
               &=\left \lfloor \frac{\nu_{k+1}(G) + \nu_{k-1}(G) }{2}\right \rfloor.
\end{align*}

\medskip

The considered cases imply that all vertices of $G$ with degree at least $3$ lie on $C$. If there is a vertex $x$ of $G$ lying on $C$ with $d_G(x)\geq k+2$, then $G$ can be represented as on Figure \ref{on_cycle_k+2}.

\begin{figure}[ht]
\centering
\begin{minipage}[b]{.5\textwidth}
  \begin{center}
  \hspace*{-4cm}
\begin{tikzpicture}[line cap=round,line join=round,>=triangle 45,x=1.0cm,y=1.0cm]
\clip(-8.,0.) rectangle (8.585264473529877,6.);
\draw(0.,2.) circle (1.6991762710207556cm);
\draw (0.9790737243720234,3.388745708272143)-- (0.92,4.36);
\draw (0.9790737243720234,3.388745708272143)-- (1.265597748851844,4.326077306776009);
\draw (0.9790737243720234,3.388745708272143)-- (1.9000477859370812,3.9528714026082232);
\draw (0.9790737243720234,3.388745708272143)-- (1.7234674181992415,4.10394249946792);
\draw (1.2856697013859437,5.390201303851337) node[anchor=north west] {$\geq k$};
\draw (0.6163658827366432,3.2) node[anchor=north west] {$x$};
\begin{scriptsize}
\draw [fill=black] (0.9790737243720234,3.388745708272143) circle (2.5pt);
\draw [fill=black] (0.92,4.36) circle (1.5pt);
\draw [fill=black] (1.265597748851844,4.326077306776009) circle (1.5pt);
\draw [fill=black] (1.9000477859370812,3.9528714026082232) circle (1.5pt);
\draw [fill=black] (1.4037561239949286,4.289375050106422) circle (0.5pt);
\draw [fill=black] (1.4996695122562225,4.244615468917818) circle (0.5pt);
\draw [fill=black] (1.582794448749344,4.187067435961041) circle (0.5pt);
\draw [fill=black] (1.7234674181992415,4.10394249946792) circle (1.5pt);
\end{scriptsize}
\end{tikzpicture}
		\end{center}
		\caption{$d_G(x) \geq k+2$}
		\label{on_cycle_k+2}
\end{minipage}%
\begin{minipage}[b]{.5\textwidth}
  \begin{center}
  \hspace*{-4cm}
\begin{tikzpicture}[line cap=round,line join=round,>=triangle 45,x=1.0cm,y=1.0cm]
\clip(-8.,0.) rectangle (8.585264473529877,6.);
\draw(0.,2.) circle (1.6991762710207556cm);
\draw (0.9790737243720234,3.388745708272143)-- (0.92,4.36);
\draw (0.9790737243720234,3.388745708272143)-- (1.265597748851844,4.326077306776009);
\draw (0.9790737243720234,3.388745708272143)-- (1.9000477859370812,3.9528714026082232);
\draw (0.9790737243720234,3.388745708272143)-- (1.7234674181992415,4.10394249946792);
\draw (1.2856697013859437,5.2) node[anchor=north west] {$\leq k-1$};
\draw (0.6163658827366432,3.2) node[anchor=north west] {$x$};
\begin{scriptsize}
\draw [fill=black] (0.9790737243720234,3.388745708272143) circle (2.5pt);
\draw [fill=black] (0.92,4.36) circle (1.5pt);
\draw [fill=black] (1.265597748851844,4.326077306776009) circle (1.5pt);
\draw [fill=black] (1.9000477859370812,3.9528714026082232) circle (1.5pt);
\draw [fill=black] (1.4037561239949286,4.289375050106422) circle (0.5pt);
\draw [fill=black] (1.4996695122562225,4.244615468917818) circle (0.5pt);
\draw [fill=black] (1.582794448749344,4.187067435961041) circle (0.5pt);
\draw [fill=black] (1.7234674181992415,4.10394249946792) circle (1.5pt);
\end{scriptsize}
\end{tikzpicture}
		\end{center}
		\caption{$d_G(x) \leq k+1$}
		\label{on_cycle_k+1}
\end{minipage}
\end{figure}
	
Observe that there is an edge $e$ of $C$ that is incident to $x$ and $\nu_{k+1}(G)=\nu_{k+1}(G-e)$. Moreover, for any edge $f$ of $C$ that is incident to $x$, $\nu_{i}(G)=\nu_{i}(G-f)$ for $i=k-1,k$. Hence we have:

\begin{align*}
    \nu_{k}(G) &=\nu_k(G-e)\\
               & \geq \left \lfloor \frac{\nu_{k+1}(G-e)  + \nu_{k-1}(G-e) }{2}\right \rfloor\\
               &=\left \lfloor \frac{\nu_{k+1}(G) + \nu_{k-1}(G)}{2}\right \rfloor.
\end{align*} Here the inequality follows from induction hypothesis applied to the components of $G-e$.

\medskip

Thus, we can assume that for any vertex $x$ of $G$ lying on $C$, we have $d_G(x)\leq k+1$. Then $G$ can be represented as on Figure \ref{on_cycle_k+1}.

Let us show that $\nu_{k+1}(G)=|E(G)|$, that is, $G$ is $(k+1)$-edge-colorable. Consider the colors $\{1,2,...,k,k+1\}$. Color the edges of the cycle $C$ with colors $1,2,3$. Observe that at each vertex of $C$ only $2$ colors will be present. Hence at each vertex of $C$ there will be missing $k-1$ colors. Since each vertex $x$ of $C$ is adjacent to at most $k-1$ vertices lying outside $C$, we can extend the edge-coloring of $C$, to a $(k+1)$-edge-coloring of $G$.

Define $x_{k-1}$ and $x_k$ as the minimum number of edges of $C$, that one needs to remove from $G$ in order to obtain a $(k-1)$ or $k$-edge-colorable subgraph of $G$, respectively. We have:
\[\nu_{k-1}(G)=|E(G)|-x_{k-1},\]
\[\nu_{k}(G)=|E(G)|-x_{k},\]
\[\nu_{k+1}(G)=|E(G)|.\]
Observe that the inequality
\[\nu_{k}(G) \geq \left \lfloor \frac{\nu_{k+1}(G) + \nu_{k-1}(G)}{2}\right \rfloor\]
that we need to prove, is equivalent to
\begin{equation}\label{xkinequality}
x_{k}\leq \left \lceil \frac{x_{k-1}}{2} \right \rceil.
\end{equation}
Let us show that the latter inequality is true. Let $J_{k-1}$ be a subgraph of $C$, such that $G-E(J_{k-1})$ is $(k-1)$-edge-colorable and $|E(J_{k-1})|=x_{k-1}$. Observe that $\Delta(J_{k-1})\leq 2$, hence \[\nu_1(J_{k-1})\geq \left \lfloor \frac{|E(J_{k-1})|}{2} \right \rfloor=\left \lfloor \frac{x_{k-1}}{2} \right \rfloor.\] Let $M_{k-1}$ be a maximum matching of $J_{k-1}$. Then $G-(E(J_{k-1})\backslash M_{k-1})$ is $k$-edge-colorable, hence
\begin{align*}
x_k &\leq |E(J_{k-1})\backslash M_{k-1}|\\
    &\leq \left \lceil \frac{|E(J_{k-1})|}{2} \right \rceil\\
    &=\left \lceil \frac{x_{k-1}}{2} \right \rceil.
\end{align*}

The proof of the theorem is complete.
\end{proof}

\begin{remark}\label{rem:TreeUniLowBound} For any $k\geq 2$, there is an infinite sequence of graphs $G$ containing $1$ cycle, such that 
\begin{eqnarray*}
	\nu_{k}(G) = \left \lfloor \frac{\nu_{k-1}(G) + \nu_{k+1}(G)}{2} \right \rfloor.
\end{eqnarray*}
\end{remark}

\begin{proof} Let $k\geq 2$ be a fixed integer. For a positive integer $l\geq 2$ consider the graph $G$ from Figure \ref{ditoxutyun}. $G$ contains one cycle of length $l$. Every vertex lying on that cycle (denoted by $C_l$) is of degree $k+1$. It is incident to 2 edges lying on the cycle and $k-1$ other edges, whose the other endpoints are degree of 1.

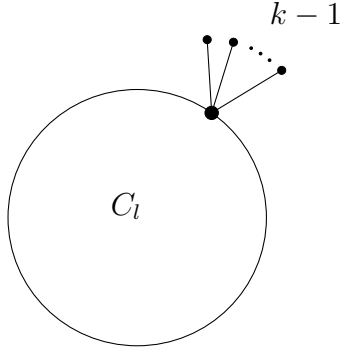
\begin{figure}[ht]
		\begin{center}
			\begin{tikzpicture}[line cap=round,line join=round,>=triangle 45,x=1.0cm,y=1.0cm]
\clip(-7.5,0.) rectangle (10.079941980215155,5.);
\draw(0.,2.) circle (1.6991762710207556cm);
\draw (0.9790737243720234,3.388745708272143)-- (0.92,4.36);
\draw (0.9790737243720234,3.388745708272143)-- (1.265597748851844,4.326077306776009);
\draw (0.9790737243720234,3.388745708272143)-- (1.9000477859370812,3.9528714026082232);
\draw (1.6,5) node[anchor=north west] {$k-1$};
\draw (-0.49379235764686163,2.472943618042581) node[anchor=north west] {$C_{l}$};
\begin{scriptsize}
\draw [fill=black] (0.9790737243720234,3.388745708272143) circle (2.5pt);
\draw [fill=black] (0.92,4.36) circle (1.5pt);
\draw [fill=black] (1.265597748851844,4.326077306776009) circle (1.5pt);
\draw [fill=black] (1.9000477859370812,3.9528714026082232) circle (1.5pt);
\draw [fill=black] (1.498775307648427,4.247061067760704) circle (0.5pt);
\draw [fill=black] (1.6222035733665614,4.175061246091793) circle (0.5pt);
\draw [fill=black] (1.7404889946797735,4.087632891208114) circle (0.5pt);
\end{scriptsize}
\end{tikzpicture}
		\end{center}
		\caption{The infinite sequence of graphs.}
		\label{ditoxutyun}
	\end{figure}
	
It can be easily checked that
\begin{align*}
	\nu_{k-1}(G) &=l\cdot (k-1), \\
	\nu_{k}(G) &=l\cdot (k-1)+\left \lfloor \frac{l}{2} \right \rfloor,\\ \nu_{k+1}(G) &=|E(G)|=l\cdot (k-1)+l,
\end{align*} hence
\begin{eqnarray*}
	\nu_{k}(G) = \left \lfloor \frac{\nu_{k-1}(G) + \nu_{k+1}(G)}{2} \right \rfloor.
\end{eqnarray*}

The proof of the remark is complete.
\end{proof}

Our next theorem verifies Conjecture \ref{BipkConj} for bipartite graphs with at most $1$ cycle.

\begin{theorem}\label{thm:BipTreeUniLowBound} For any $k\geq 2$ and a bipartite graph $G$ containing at most $1$ cycle, 
\begin{eqnarray*}
	\nu_{k}(G) \geq \frac{\nu_{k-1}(G) + \nu_{k+1}(G)}{2}.
\end{eqnarray*}
\end{theorem}

\begin{proof} The proof of this theorem is identical to that of Theorem \ref{thm:TreeUniLowBound}, with the exception that inequality (\ref{xkinequality}) should be replaced with
\begin{equation*}
x_{k}\leq \frac{x_{k-1}}{2}.
\end{equation*}The latter can be proved in a similar way, by taking into account that if $C$ is not the odd cycle, then 
\[\nu_1(J_{k-1})\geq \frac{|E(J_{k-1})|}{2}= \frac{x_{k-1}}{2}.\]
The proof of the theorem is complete.
\end{proof}

\begin{corollary}\label{TreeCorollary} For any $k\geq 2$ and a tree $T$
\begin{eqnarray*}
	\nu_{k}(T) \geq \frac{\nu_{k-1}(T) + \nu_{k+1}(T)}{2}.
\end{eqnarray*}
\end{corollary}

Combined with the classical theorem of K\"onig, Corollary \ref{TreeCorollary} implies:
\begin{corollary}\label{TreePerfect3Corollary} If $T$ is a tree containing a perfect matching and $\Delta(T)=3$, then 
\begin{eqnarray*}
	\nu_{2}(T) \geq \frac{3n-2}{4}.
\end{eqnarray*}
\end{corollary}

\begin{remark}\label{rem:BipTreeUniLowBound} For any $k\geq 2$, there is an infinite sequence of bipartite graphs $G$ containing $1$ cycle, such that 
\begin{eqnarray*}
	\nu_{k}(G) = \frac{\nu_{k-1}(G) + \nu_{k+1}(G)}{2}.
\end{eqnarray*}
\end{remark}

\begin{proof} Consider the sequence of graphs $G$ from Remark \ref{rem:TreeUniLowBound} when $l$ is even. Observe that
\begin{eqnarray*}
	\nu_{k}(G) = \frac{\nu_{k-1}(G) + \nu_{k+1}(G)}{2}.
\end{eqnarray*}
The proof of the remark is complete.
\end{proof}




\nocite{*}


\bibliographystyle{abbrv}


\end{document}